\documentclass[a4paper,onecolumn,12pt,notitlepage,tightenlines]{revtex4-1}
\usepackage{amssymb}
\usepackage{graphicx}
\usepackage{hyperref}
\usepackage{color}
\usepackage{tikz}

\usepackage[T1]{fontenc}
\usepackage[latin9]{inputenc}

\usepackage{amsthm}
\usepackage{bbm}
\usepackage{mathtools}
\newtheorem{thm}{Theorem}[section]
\newtheorem{defi}[thm]{Definition}
\newtheorem{lem}[thm]{Lemma}
\newtheorem{prop}[thm]{Proposition}

\newtheorem{problem}[thm]{Problem}

\def\tr{\operatorname{tr}}
\def\idty{{\leavevmode\rm 1\mkern -5.4mu I}} 

\def\Rl{{\mathbb R}}\def\Cx{{\mathbb C}}
\def\Ir{{\mathbb Z}}\def\Nl{{\mathbb N}}

\def\norm #1{\Vert #1\Vert}

\def\braket#1#2{\langle #1,#2\rangle}
\def\brAket#1#2{\langle #1\vert#2\rangle}
\def\brAAket#1#2#3{\langle#1\vert#2\vert#3\rangle}
\def\bra#1{{\langle#1\vert}}
\def\ket #1{\vert#1\rangle}
\def\ketbra #1#2{{\vert#1\rangle\!\langle#2\vert}}
\def\kettbra#1{\ketbra{#1}{#1}}

\def\abs#1{\vert#1\vert}

\def\Order{{\bf O}}\def\order{{\bf o}}

\def\inv{^{-1}}
\def\erfc{\mbox{\rm erfc}}

\def\re{\Re e}

\def\BB{{\mathfrak B}}\def\HH{{\mathfrak H}}

\def\XX{{\mathcal X}}
\def\gen{{\mathcal L}}  
\def\genh{{\widehat\gen}}   \def\rhoh{{\widehat\rho}}
\def\semg{{\mathcal T}} 
\def\dom{\mathop{\rm dom}\nolimits}

\def\res{{\mathcal R}}  
\def\resl{\res_\lambda}
\def\id{{\rm id}}

\def\Exit{{\mathcal E}}

\def\upromantwo{^{\mbox{\small\textnormal{I\!I}}}}
\def\upromantwo{^{\scriptstyle\rangle\!\langle}}
\def\jsq{j\upromantwo}
\def\domsq#1{(\dom#1)\upromantwo}

\def\LL{L} 

\def\ppert{{\mathcal P}} 
\def\tcop{{\mathcal S}}  
\def\reins{{\mathcal S}} 
\def\transit{{\mathfrak N}} 

\def\id{{\mathcal I}}

\def\jt{\tilde\jmath}

\def\tc{{\mathfrak T}} 
\def\th{\tc(\HH)}
\def\bh{\BB(\HH)}
\begin{document}
\title{Unbounded generators of dynamical semigroups\footnote{Based on a talk given by R.F.W. at the 2016 Torun Conference}}
\author{I. Siemon$^\dagger$}
\author{A. S. Holevo$^*$}
\author{R. F. Werner$^\dagger$}
\affiliation{$^\dagger$ Leibniz Universit\"at Hannover - Institut f\"ur Theoretische Physik\\
$^*$Steklov Mathematical Institute of Russian Academy of Sciences, Moscow\\
and National Research University Higher School of Economics}

\begin{abstract}
Dynamical semigroups have become the key structure for describing open system dynamics in all of physics. Bounded generators are known to be of a standard form, due to Gorini, Kossakowski, Sudarshan and Lindblad. This form is often used also in the unbounded case, but rather little is known about the general form of unbounded generators. In this paper we first give a precise description of the standard form in the unbounded case, emphasizing intuition, and collecting and even proving the basic results around it. We also give a cautionary example showing that the standard form must not be read too naively. Further examples are given of semigroups, which appear to be probability preserving to first order, but are not for finite times. Based on these, we construct examples of generators which are not of standard form.
\end{abstract}

\maketitle
\section{Introduction}
Dynamical semigroups have become the key structure for describing open system dynamics in all of physics. Their importance for describing processes with decoherence can hardly be underestimated, and with the new push towards quantum technologies, where noise is the principal enemy, their role has been steadily growing. Yet our structural understanding of dynamical semigroups is curiously limited. For comparison look at the case of reversible dynamics: In that case we know that any time evolution with continuous expectation values is implemented by a continuous unitary group, which is in turn generated by a self-adjoint Hamiltonian operator. So we have a complete mathematical characterization of all such evolutions. In fact, the spectral theorem for unbounded self-adjoint operators was one of the first elements of the mathematical structure of quantum mechanics which von Neumann developed, and he did it for just this purpose. The analogous open systems problem then would be the following:

\begin{problem}
Consider a Hilbert space $\HH$. Characterize all one-parameter semigroups $t\mapsto\semg_t$, ($t\geq0$) such that each $\semg_t$ is a completely positive map on the trace class $\th$, and, for any $\rho\in\th$ and any bounded operator $A\in\bh$, we have $\lim_{t\to0}\tr\bigl(\semg_t(\rho)A\bigr)=\tr(\rho A)$.
\end{problem}

Of course, this volume celebrates the solution of a fundamental special case of this problem, namely the case of bounded generators, which is equivalent to the uniform continuity condition $\lim_{t\to0}\norm{\semg_t-\id}=0$. But the problem as written above, i.e., of characterizing also the merely strongly continuous dynamical semigroups or, equivalently, their unbounded generators, is open.
In spite of this, many applications use unbounded versions of the GKLS-form of the generator, which we will call the ``{\it standard form}\/'' in the sequel. The basic idea for the standard form goes back to Davies \cite{DaviesN}, and has been somewhat further developed since \cite{FagnolaProy,HolevoCov,Chebotarev1991}. The typical attitude towards this problem is currently to use unbounded standard forms where it seems natural, but to avoid the general unbounded case. In fact, at the Torun conference one prominent member of our community called that a hopeless problem. Indicative of this state of affairs is that the papers \cite{HolevoExc,HolevoExa} from 1995/96, which presents an example of a non-standard generator has practically not been cited. Likewise underrated is the work of Bill Arveson \cite{ArvesonB}, which also goes well beyond the standard form. The last author is indebted to Franco Fagnola for reminding him of this work, which he earlier had ignored erroneously as being mostly about the special case of endomorphism semigroups \cite{ArvesonGen}.

Therefore, our aim in this survey is twofold: Firstly, we will give a description of the standard form in the unbounded case, emphasizing intuition, and collecting and even proving the basic results around it. We also give a cautionary example (Sect.~\ref{sec:NonClos}) showing that the standard form must not be read too naively. Further examples are given of semigroups \cite{DaviesN,HolevoExa}, which appear to be probability preserving to first order (i.e., when looking only at the generator on the finite-rank part of its domain), but not for finite times. This phenomenon is akin to classical processes allowing escape to infinity in finite time. Secondly, we will give examples of generators which are {\it not} of standard form, by modifying the previous examples.

In order to see what kind of characterization of generators might be hoped for, it is helpful to look for guidance in the classical case. For example, we could replace $\th$ by $\LL^1([0,1])$, the integrable functions on the unit interval, and ask, similarly, for all continuous Markov semigroups on this space. However, the structure of $\LL^1([0,1])$ as an ordered Banach space is identical to the integrable functions on any atomless measure space, like $\Rl$ of $\Rl^d$. So part of the answer would be diffusions on $\Rl^d$, recoded in some way to the unit interval (mapping measure zero sets to measure zero sets).  In particular, not even the dimension $d$ of the underlying space can be seen in the characterization. As a consequence, the classification is likely to be wild and uninformative, unless further structure is imposed. An example would be the Feller condition, demanding that the dynamics in the Heisenberg picture takes continuous functions to continuous functions. Clearly, the algebra of continuous functions is sensitive to dimension, and a fruitful theory becomes possible.

Consider, instead, to  replace $\th$ by a sequence space like $\ell^1(\Ir)$. Then one expects a characterization of the generators in terms of transition rates, and there is some well-developed theory around this \cite{Kemeney}. The quantum case is now somewhere between these two classical examples: On one hand, the trace class has some discreteness like $\ell^1(\Ir)$, because it contains many pure states, and $\bh$ has many minimal projections. On the other hand, there is a continuum of pure states, somewhat reminiscent of the continuum of points in the interval $[0,1]$.

This classical comparison sets the theme for our survey: We will look especially at the pure states in the domain of the generator. This will at the same time give us a useful definition of the standard form and point to the possibility for non-standard generators. We will give rigorous statements throughout, while emphasizing intuition. For generalities on semigroups on Banach spaces we recommend \cite{Nagel}. To fix our setting, we assume throughout that $\HH$ is a separable Hilbert space over $\Cx$. A {\it dynamical semigroup} is just a one-parameter semigroup as described in the above Problem, of which we assume that $\tr\semg_t(\rho)\leq\tr\rho$ for all $0\leq\rho\in\th$. If equality holds here for all $\rho$, we call $\semg_t$  {\it conservative}. Throughout, we will denote the domain of an unbounded operator $\gen$ by $\dom\gen$. We normally work in the Schr\"odinger picture, i.e., in terms of operators $\semg$ on $\th$. Their adjoints acting on the bounded operators $\bh$, a.k.a.\ the channel in the Heisenberg picture, are then denoted by $\semg^*$, so that $\tr\semg(\rho)A= \tr\rho\semg^*(A)$.

\section{Standard generators}
\subsection{Bounded and standard generators}
Let us recall the standard (GKLS-) form of the generator, established in the bounded case. In this case $\semg_t=\exp(t\gen )$, where
\begin{eqnarray}\label{stand}
  \gen\rho&=&K\rho+\rho K^*+\sum_\alpha L_\alpha\rho L_\alpha^*\qquad\mbox{with}\\
         0&\geq&   K+ K^*+\sum_\alpha L_\alpha^* L_\alpha \label{standnorm}
\end{eqnarray}
for some bounded operators $K,L_\alpha$. The set of labels $\alpha$ may be infinite, in which case the sum in \eqref{standnorm} is taken in the weak operator topology and converges as a bounded increasing sequence, and
the sum in \eqref{stand} then converges in trace norm.

A conspicuous feature of this form is the separation into a part associated with $K$, and another which is associated with the jump operators $L_\alpha$. An intuitive way to understand this is the observation that $\exp(t\gen)$ must be a completely positive map norm close to the identity. This means \cite{StineCont} that it must also have a Stinespring dilation close to that of the identity. Now the only Kraus operator in the decomposition of the identity is the unit operator $\idty$, so one of the Kraus operators of $\exp(t\gen)$ can be chosen to be close to $\idty$, say $\approx\idty+t K$. The others will then have to scale like $\approx\sqrt t L_\alpha$, which gives
$\semg_t\rho=\rho+t\gen(\rho)+\Order(t^2)$ with the above generator. The dominant Kraus operator $(\idty+tK)$ belongs to a pure operation, i.e., an operation taking pure states into pure states \cite[Sect.~2.3]{Davies}. The only difference to the unitary case is that this part now typically loses normalization, so the evolution takes pure states to multiples of pure states.

To summarize, the generator splits into one part, which by itself generates an evolution taking pure states to pure states and a second part, which is completely positive. The work of Davies and the stochastic calculus suggest the following terminology:

\begin{defi}\label{def:no-event}
A {\bf no-event semigroup} on a Hilbert space $\HH$ is a dynamical semigroup $\semg^0_t$, $t>0$ such that every pure state $\rho=\kettbra\psi$ is mapped to a multiple of a pure state.
It is necessarily of the form $\semg^0_t\rho=C_t\rho C_t^*$ with $C_t=\exp(tK)$ a strongly continuous contraction semigroup of Hilbert space operators.
\end{defi}

Note that this definition no longer requires $K$ to be bounded. Moreover, it also makes sense in the discrete classical case, i.e., for semigroups on $\ell^1(X)$ for some countable set $X$. Pure states $\delta_x$ are then of the form concentrated on a single point $x\in X$, corresponding to the probability distribution $\delta_x(y)=\delta_{x,y}$. It is easy to see that a no-event semigroup cannot change $x$, i.e., it must be of the form
\begin{equation}\label{noeventclass}
  \bigl(\semg^0_t\bigr)(\delta_x)=e^{-t\mu_x}\,\delta_x,
\end{equation}
where $\mu:X\to\Rl^+$ describes the loss rate from state $x$. The function $\mu$ need not be bounded. Just as in the quantum case, the whole generator will differ from the no-event part by a positive term, which describes the rates of transitions from $x$ to other states $y$, resulting in the usual rate matrix.

The basic idea of constructing the generator (classical or quantum) is that the positive term in the generator will make the semigroup more nearly conservative, i.e., it will compensate some of the normalization loss in $\semg^0_t$. But, due to the overall (sub-)normalization condition $\tr\semg_t(\rho)\leq\tr\rho$, there cannot be more transitions than there is loss. This means the positive part must be {\it bounded with respect to the normalization loss} of the no-event part. Thus all unboundedness is tamed, once it is under control for the no-event part. Of course, in principle, there might not be such a no-event part in the generator. But for the moment, we define the ``good'' case by this property:

\begin{defi}\label{def:standard}
A dynamical semigroup is called {\bf standard} if it is the minimal solution arising from a completely positive perturbation of the generator of a no-event semigroup.
\end{defi}

We have not yet defined ``the minimal solution'' in this sentence, and this will be the task of Sect.~\ref{sec:minsol}. Standard generators look just like \eqref{stand}, with the following changes:
$K$ is the generator of an arbitrary contraction semigroup on $\HH$, and the jump operators need to be operators
\begin{equation}\label{standUB}
   L_\alpha:\dom K\to\HH  \quad\text{with}\ \sum_\alpha\norm{L_\alpha\phi}^2\leq -2\re\braket\phi{K\phi}.
\end{equation}
The generator is thus naturally split into $\gen=\gen^0+\ppert$, i.e., no-event part and completely positive perturbation,  namely
\begin{equation}\label{ppertLind}
  \gen^0(\ketbra\phi\psi)= \ketbra{K\phi}{\psi}+\ketbra{\phi}{K\psi} \quad \text{and}\quad
  \ppert(\ketbra\phi\psi)=\sum_\alpha\ketbra{L_\alpha\phi}{L_\alpha\psi}.
\end{equation}
The natural domain for all these operators is  $\domsq K$, defined as the set of finite linear combinations of rank 1 operators $\ketbra\phi\psi$ with $\phi,\psi\in\dom K$. In particular, the expression for $\ppert$ does not require the adjoint $L_\alpha^*$ to be even defined, which is important because it might not exist (see Sec.~\ref{sec:NonClos} below). The effect of the minimal solution construction is then to extend the domain of $\gen$ beyond $\domsq K$, so that in the end we may well get some  $\rho\in\dom\gen$, for which the individual terms $\gen^0\rho$ and $\ppert\rho$ are no longer well defined.

\subsection{Exit spaces and reinsertions}
In this section we will give a dynamical interpretation of the standard form, which forms the background for the term ``no-event'' semigroup. This interpretation is consistent also with the unbounded standard form. It provides the basis for the more technical statement that, for a standard generator, all the unboundedness is already determined by the no-event part, relative to which the positive perturbation $\ppert$ is bounded. This section provides some background, and is not needed to understand the later sections.

The idea behind the term ``no-event semigroup'' is that it describes the evolution for as long the system has not yet been captured, i.e., up until a detection or ``arrival'' event \cite{RFWarrival,HolevoExc}. Modifying a Hamiltonian by absorbing terms $-iK$ with $K\geq0$ is, in fact, one of the standard ways to describe a detection process. By choosing $K$ to be spatially localized in a region, we get a model of a detector in that region. The probability for detection in the time interval $[t,s]$, starting from an initially normalized state $\rho$ is then, by definition $\tr\semg^0_t\rho-\tr\semg^0_s\rho$. Clearly, this defines a POVM for the arrival time distribution, which also allows for the possibility that the particle never arrives.  We would also like to find the observables which are jointly measurable with arrival. For example, when there are several detectors, we need to know which of them fired. This is naturally captured by the notion of the {\it exit space} of a contraction semigroup \cite{RFWarrival}. For a semigroup $e^{tK}$ we consider the normalization loss as a quadratic form on $\dom K$, and define an exit space for $K$ as a pair $(\Exit,j)$ of a Hilbert space $\Exit$ and a linear map $j:\dom K\to \Exit$ such that, for $\psi,\phi\in\dom K$,
\begin{equation}\label{exitspace}
  \braket{j\psi}{j\phi}=-\frac d{dt}\Bigl\langle e^{tK}\psi\Bigm\vert e^{tK}\phi\Bigr\rangle
      =-\bigl(\braket{K\psi}{\phi}+\braket{\psi}{K\phi}\bigr).
\end{equation}
There is always a unique minimal exit space: The separated completion of $\dom K$ with respect to the above scalar product. However, for reasons which will be apparent later, we also allow non-minimal exit spaces, possibly even with an inequality $\leq$ instead of equality in \eqref{exitspace}.

Now if $F\in\BB(\Exit)$ is an effect operator describing some yes-no-question asked at exit time, we set the probability density for obtaining that result at time $t$,  on an initial preparation $\kettbra\phi$ with $\phi\in\dom K$, to be $\brAAket{je^{tK}\phi}F{je^{tK}\phi}$. More formally, we consider a map
$J:\HH\to\LL^2(\Rl_+,dt;\Exit)$. The range of $J$ is the space of $\Exit$-valued functions on $\Rl_+$, which is canonically isomorphic to $\LL^2(\Rl_+,dt)\otimes\Exit$, but the function notation is more helpful for our purpose. We set $(J\phi)(t)=j\bigl(e^{tK}\phi\bigr)\in\Exit$ for $\phi\in\dom K$. Then $J$ extends to $\HH$ by continuity, because
\begin{eqnarray}
\norm{J\phi}^2&=&\int_0^\infty\!\!dt\ \norm{je^{tK}\phi}_\Exit^2=-\int_0^\infty\!\!dt\ \frac d{dt}\norm{e^{tK}\phi}^2\nonumber\\
              &=&\norm\phi^2-\lim_{t\to\infty}\norm{e^{tK}\phi}^2
              \leq\norm\phi^2  \label{Jbounded}
\end{eqnarray}
The joint probability for an $F$-detection in the time interval $[t,s]$ on the initial state $\rho$ is then
\begin{equation}\label{exitprob}
  \tr\bigl(\rho J^* (\chi_{[t,s]}\otimes F)J\bigr)=\int_t^s\!\!d\tau\ \brAAket{je^{\tau K}\phi}F{je^{\tau K}\phi}.
\end{equation}
Here the right hand side just uses the density mentioned above for $\phi\in\dom K$, and the left hand side is the same for $\rho=\kettbra\phi$, but makes sense for arbitrary $\rho$ by virtue of the continuous extension.

We can turn the arrival time detection into a dynamical, repeatable  process  on $\HH$ by introducing a {\it reinsertion map}, which transforms the ``state upon exit'' into a new state of the system. This is done by a completely positive, trace non-increasing map $\reins:\tc(\Exit)\to\th$. Then the effect $F$ in \eqref{exitprob} may arise from a measurement on the original system, including an arrival time measurement of just the same kind.

Before iterating this idea, let us simplify the description by introducing the Stinespring dilation, i.e., a contraction
$v:\Exit\to\transit\otimes\HH$, so that $\reins(\sigma)=\tr_\transit v\sigma v^*$. Observables on $\transit$ then describe the information that can be extracted at the moment of a jump, so we call $\transit$ the {\it transit space}. Composing $v$ with $j$ we get a map $\jt=vj:\dom K\to\transit\otimes\HH$, which, apart from the special form of the image space satisfies exactly the requirements \eqref{exitspace} for an exit space (possibly with an inequality, if $\reins$ can reduce the trace). In this sense a process of exit and reinsertion is completely specified by an exit space of the special form $(\transit\otimes\HH,\jt)$.

From now on we will take $J$ to be defined by $\jt$.  We can iterate this operator to a sequence of maps $J^{(n)}:\HH\to\bigl(\LL^2(\Rl_+,dt;\Exit)\bigr)^{\otimes n}\otimes\HH$, with $J^{(0)}=\idty_\HH$, $J^{(1)}=J$, and $J^{(n+1)}=(\idty^{\otimes n}\otimes J)J^{(n)}$. This has the same interpretation as $J$, only that we are now looking at $n$ consecutive events. The $n$ time arguments of wave functions in this space are the time increments between successive events. In order to get a dynamical semigroup out of this iteration, we need to fix a time interval $[0,\tau]$ and look only at events happening during this interval. We also need to evolve the system up to time $\tau$ after the last event with a further application of the no-event semigroup. Thus we set $J^{(n)}_\tau$ to be a map between the same spaces as $J^{(n)}$, but modified as
\begin{equation}\label{Jntau}
 (J^{(n)}_\tau\phi)(t_1,\ldots,t_n)
       =(\idty^{\otimes n}\otimes e^{(\tau-\sum_i t_i)K}) (J^{(n)}\phi)(t_1,\ldots,t_n), \end{equation}
whenever $\sum_it_i\leq\tau$, and zero otherwise. So $J^{(n)}_\tau$ is a dilation of the evolution conditional on exactly $n$ events happening in that interval. The conditional evolution up to the end of this interval is $\semg^{(n)}_\tau\rho=\tr_{\rm events} J^{(n)}_\tau\rho J^{(n)\,*}_\tau$, where the trace is the partial trace over the tensor factor $\bigl(\LL^2(\Rl_+,dt;\Exit)\bigr)^{\otimes n}$.
Then $\semg_\tau\rho=\sum_{n=0}^\infty\semg^{(n)}_\tau$ is a dynamical semigroup. In fact, it is the same minimal semigroup as constructed in the next section.
We will not go through the proof of this assertion, which is best done via the Laplace transforms of the $\semg^{(n)}_\tau$, which turn out to be the exactly the terms in the sequence \eqref{resmin} below.

Experts in the stochastic calculus will easily recognize the dilation construction here. In fact, when we write the time arguments in the space $\bigl(\LL^2(\Rl_+,dt;\Exit)\bigr)^{\otimes n}$ not as increments
but as the absolute event times $\tau_i=\sum_{k=1}^i t_i$, we get wave functions defined on ordered time arguments, which have unique symmetric and antisymmetric extensions to arbitrary $n$ tuples of times, yielding the Fermionic and the Bosonic stochastic integrals. Our focus here was just the dynamical semigroup, however, and specifically to trace the implications of unboundedness through the construction. Indeed the turning point is \eqref{Jbounded}: Once $J$ has been extended from $\dom K$ to a bounded operator on all of $\HH$, the entire further construction is in terms of  bounded operators, and no more domain questions need to be addressed.

The exit\&reinsertion picture suggests other standard ways to look at the generator, which are brought together with the form \eqref{standUB} in the following proposition. It also lists (in ({\it d\/})) the form we prefer for the next section. For the action of the exit space injection $j$ on mixed states we introduce the linear operator $\jsq:\domsq K\to\tc(\Exit)$ given by
$\jsq(\ketbra\phi\psi)=\ketbra{j\phi}{j\psi}$. Then we have:

\begin{prop}\label{prop:stdExit}
Let $t\mapsto\exp(tK)$ be a contraction semigroup on $\HH$ with generator $K$ and minimal exit space $(\Exit,j)$.
Then standard generators with no-event semigroup $\semg^0_t\rho=e^{tK}\rho e^{tK ^*}$ are equivalently characterized by any of the following sets:
\begin{itemize}
\item[(a)] Completely positive ``reinsertion'' maps $\reins:\tc(\Exit)\to\th$ with\\ $\tr\reins(\sigma)\leq\tr\sigma$.
\item[(b)] Non-minimal exit spaces of product form, i.e., maps $\jt:\dom K\to\transit\otimes\HH$ such that $\norm{\jt\phi}^2\leq 2\re\braket\phi{K\phi}$.
\item[(c)] Maps $\ppert:\domsq K\to\th$, which can be  written in the form \eqref{ppertLind} with jump operators $L_\alpha$ satisfying \eqref{standUB}.
\item[(d)] Completely positive maps $\ppert:\dom\gen^0\to\th$, with $\tr\ppert\rho\leq-\tr\gen^0\rho$ \\for all positive $\rho\in\dom\gen^0$.
\end{itemize}
The correspondence is given by restriction from (d) to (c), and by unique $\gen^0$-graph-norm continuous extension in the other direction.
Between (a),(b),(c) it is given on $\domsq K$ by $\ppert=\reins\jsq=(\id\otimes \tr_\transit)\jt\upromantwo$.
Possible choices of jump operators correspond precisely to choices of Kraus operators for $\reins$ or a basis $e_\alpha\in\transit$,
with $\reins(\sigma)=\sum_\alpha M_\alpha\sigma M_\alpha^*$, via $L_\alpha=M_\alpha j$
and $\jt\phi=\sum_\alpha e_\alpha\otimes L_\alpha\phi$.
\end{prop}

\begin{proof}
The equivalences are largely trivial to verify on $\domsq K$, or have already been described in the text above. The only statement not of this kind is the continuous extension $(c){\to}(d)$.
Here we note that $\domsq K\subset\dom\gen^0$ is invariant and dense in $\th$, hence a core, so that continuity will guarantee an extension to $\dom\gen^0$.
Since $\reins$ is clearly trace norm continuous, the identity $\ppert=\reins\jsq$ shows that we only need to prove the continuity of $\jsq$, i.e., the statement that $\jsq\rho_n\to0$, whenever $\rho_n\to0$ and $\gen^0\rho_n\to0$ (each limit in trace norm). We will do this by establishing the estimate $\norm{\jsq\rho}\leq \norm{\gen^0\rho}$.

By definition of $\domsq K$, we can write $\rho=\sum_\ell^N r_\ell \ketbra{\phi_\ell}{\psi_\ell}$ with $r_\ell\in\Cx$ and $\phi_\ell,\psi_\ell\in\dom K$. Now on the finite dimensional span of the $\phi_\ell,\psi_\ell$ we can perform a singular value decomposition and get a more canonical form of $\rho$, where $r_\ell>0$, and each of the families $\{\phi_\ell\},\{\psi_\ell\}$ is orthonormal. Then we have
\begin{eqnarray}
  \norm{\jsq\rho}&=&\Bigl\Vert\sum_\ell r_\ell \ketbra{j\phi_\ell}{j\psi_\ell}\Bigr\Vert
      \leq\sum_\ell r_\ell \norm{j\phi_\ell}\norm{j\psi_\ell}\nonumber\\
     &\leq&\sum_\ell \frac{r_\ell}2 \bigl(\norm{j\phi_\ell}^2+\norm{j\psi_\ell}^2\bigr)
      =-\re\sum_\ell {r_\ell} \bigl(\braket{\phi_\ell}{K\phi_\ell}+\braket{K\psi_\ell}{\psi_\ell}\bigr) \nonumber\\
     &=&-\re\sum_{\ell,m} {r_\ell} \bigl(\braket{\phi_m}{K\phi_\ell}\braket{\psi_\ell}{\psi_m}
                                         +\braket{\phi_m}{\phi_\ell}\braket{K\psi_\ell}{\psi_m}\bigr) \nonumber\\
     &=&-\re\tr W\gen^0\rho,
\end{eqnarray}
where $W=\sum_m\ketbra{\psi_m}{\phi_m}$. This is a partial isometry, so $\norm W=1$, and hence $\norm{\jsq\rho}\leq\norm{\gen^0\rho}$.
\end{proof}

\subsection{The minimal solution}\label{sec:minsol}
Adding a further term (``a perturbation'') to a well-known ``simple'' generator is, of course, commonplace throughout quantum mechanics and more general evolution equations. Very often one considers perturbations which are relatively bounded with respect to the given generator. In this case \cite{RSimon} the domain of the perturbed generator remains the same. The perturbations considered here will usually {\it not} be of this kind. There are two equivalent versions of the construction. One is based on the resolvent series \cite{DaviesN}, and one on the iteration of integral equations \cite{HolevoExc}. Since the resolvent version can be stated slightly more compactly, and we will need to consider resolvents anyhow, we will choose this version.

The resolvent of a semigroup $\semg_t=\exp(t\gen)$ is given, for any $\lambda>0$, by the integral
\begin{equation}\label{resgen}
  \resl=(\lambda-\gen)\inv=\int_0^\infty\!dt\ e^{-\lambda t}\semg_t.
\end{equation}
From this definition it is clear that $\resl$ is completely positive, and satisfies the norm bound $\norm{\lambda\resl}\leq 1$, and the resolvent identity $\resl-\res_\mu=(\mu-\lambda)\resl\res_\mu$. Conversely, any family of operators $\resl:\th\to\th$ satisfying these conditions defines a dynamical semigroup, which can be recovered by the formula
\begin{equation}\label{res2sem}
  \semg_t=\lim_n\bigl(1-\frac tn\gen\bigr)^{-n}=\lim_{n\to\infty} \left(\frac nt \res_{n/t}\right)^n.
\end{equation}
Here the middle part is provided only as a formal expression to explain what the right hand side should look like (but see \cite{Nagel} for a proof). Moreover, we have, for any $\lambda>0$,
\begin{equation}\label{domL}
  \dom\gen=\resl\bigl(\th\bigr),\qquad\mbox{with}\ \gen\resl\rho=\lambda\resl\rho-\rho.
\end{equation}

Now consider a generator $\gen^0$, typically (but not always) of a no-event semigroup, from which we would like to construct a new generator $\gen=\gen^0+\ppert$ with $\ppert$  completely positive.
For the construction of standard generators the forms of $\gen^0$ and $\ppert$ are given in \eqref{ppertLind}. The domain of $\gen$ should be at least $\dom\gen^0$, and we want the normalization of the new semigroup to be non-increasing. This fixes the normalization condition \eqref{standUB}. Moreover, for $\rho\geq0$,
\begin{equation}\label{ppertnormal}
  0\geq \tr(\gen^0+\ppert)\resl^0\rho=\tr\bigl(\lambda\resl^0\rho-\rho +\ppert\resl^0\rho\bigr)
    \geq\tr\ppert\resl^0\rho-\tr\rho.
\end{equation}
Hence $\ppert\resl^0$ is everywhere defined, completely positive, and trace non-increasing.
Therefore, $\norm{\ppert\resl^0}\leq1$.

Formally, we get the resolvent $\resl$ of the perturbed semigroup from
\begin{equation}\label{resdiff}
  \resl-\resl^0=\resl\Bigl((\lambda-\gen^0)-(\lambda-\gen)\Bigl)\resl^0=\resl\ppert\resl^0.
\end{equation}
Still proceeding formally, we can use this to determine $\resl$ by iteration, or equivalently to solve the Neumann series for $(\id-\ppert\resl^0)\inv$ to find:
\begin{equation}\label{resmin}
  \resl=\sum_{n=0}^\infty \resl^0\bigl(\ppert\resl^0\bigr)^n.
\end{equation}
The basic algebra here is quite standard, and used also for the relatively bounded perturbation theory of generators. In that case $\norm{\ppert\resl^0}<1$, so the series obviously converges in norm. Moreover, one can then write the factor $\resl^0$ outside the sum, so that $\dom\gen=\resl(\th)\subset\resl^0(\th)=\dom\gen^0$, and the domain will not increase. This will be different now. We state the basic construction result without assuming that $\gen^0$ is a no-event semigroup. This is because this generalization will be needed in Sect.~\ref{sec:nonstand}. For use in that section we also provide Lemma~\ref{lem:finrank}, showing that sometimes the domain does not increase.

\begin{prop}\label{prop:minsol}
Let $\gen^0$ be the generator of a dynamical semigroup, and let $\ppert:\dom\gen^0\to\th$ be a completely positive map such that, for $0\leq\rho\in\dom\gen^0$,
\begin{equation}\label{subnPP}
  \tr\ppert(\rho)\leq -\tr\gen^0(\rho).
\end{equation}
Then $\ppert\resl^0$ is a completely positive operator on $\th$, and the series \eqref{resmin} converges strongly to the resolvent $\resl$ of a dynamical semigroup. $\XX=\resl$ is the smallest completely positive solution of the equation $\XX=\resl^0+\XX\ppert\resl^0$ in completely positive ordering, and is hence called the {\bf minimal resolvent solution} associated with the perturbation $\ppert$.
\end{prop}

\begin{proof}
We only sketch the key idea, which makes clear why the series indeed converges, even without assuming $\norm{\ppert\resl^0}<1$. The the partial sum truncated at $n$ is just the $n^{\rm th}$ iterate $\res^{(n)}_\lambda$ defined by $\res^{(0)}_\lambda=\resl^0$ and
\begin{equation}\label{ppertiterate}
  \res^{(n+1)}_\lambda=\resl^0+\res^{(n)}_\lambda\ppert\resl^0.
\end{equation}
We will prove by induction that for positive $\rho$, we have $\tr\lambda\res^{(n)}_\lambda\rho\leq\tr\rho$. Indeed, this is true for $n=0$, like for the resolvent of any dynamical semigroup and, by the induction hypothesis,
\begin{eqnarray}
  \tr\lambda\res^{(n+1)}_\lambda\rho
    &\leq& \tr\lambda\resl^0\rho + \tr\ppert\resl^0\rho    \nonumber\\
    &\leq& \tr\lambda\resl^0\rho- \tr\gen\resl^0\rho
      =\tr\lambda\resl^0\rho-\tr\bigl(\lambda\resl^0\rho-\rho\bigr)=\tr\rho.   \nonumber
\end{eqnarray}
Hence the sequence $\lambda\res^{(n)}_\lambda\rho$ is increasing and uniformly bounded in trace norm, and therefore convergent in norm. By linearity this extends to the trace class, and applying it to a matrix of trace class operators we conclude that the limit $\resl$ is a completely positive operator.

If $\tcop$ is any completely positive solution of the equation in the Proposition, we have that $(\tcop-\resl^{(0)})=(\tcop-\resl^0)$ is completely positive, and because
\begin{equation}\label{tcopiterate}
  (\tcop-\res^{(n+1)}_\lambda)=(\tcop-\res^{(n)}_\lambda)\ppert\resl
\end{equation}
this persists through the iteration, and the result follows by taking the limit.
\end{proof}

\begin{lem}\label{lem:finrank}
If, in the setting Prop.~\ref{prop:minsol}, the perturbation $\ppert$ has finite rank, we have $\dom\gen=\dom\gen^0$.
\end{lem}

\begin{proof}
We will show that, for some $n$, $\norm{(\ppert\resl^0)^n}<1$. Then the resolvent series \eqref{resmin} converges in norm, even without the factor $\resl^0$ in each term, so as argued after that equation, the domain will not increase.

By definition, a finite rank operator and its adjoint can be written as
\begin{equation}\label{finrank}
  \ppert\resl^0\rho=\sum_i\sigma_i\tr(S_i\rho)      \quad\mbox{and}\
  (\ppert\resl^0)^*X= \sum_iS_i\tr(\sigma_iX),
\end{equation}
where the sum is finite and the $\sigma_i\in\th$ and the $S_i\in\bh$ are chosen linearly independent. The action on the linear span of the $\sigma_i$ is given by the finite dimensional matrix $P_{ij}=\tr S_i\sigma_j$ in the sense that
$\ppert\resl^0\sum_jx_j\sigma_j=\sum_i(\sum_jP_{ij}x_j)\sigma_i$.

Because $\norm{\ppert\resl^0}\leq1$,  all the eigenvalues of the matrix $P$ must be in the unit circle. If there are no eigenvalues of modulus one, the powers of $P$ and hence of $\ppert\resl^0$ contract exponentially to zero, and we are done. Now suppose $P$ has an eigenvalue of modulus one. Then so does its transpose, and we hence have an operator $X$ with $(\ppert\resl^0)^*X=\omega X$ with $\abs\omega=1$. Then 2-positivity implies
\begin{equation}\label{2posinproof}
  (\ppert\resl^0)^*(X^*X)\geq (\ppert\resl^0)^*(X)^*(\ppert\resl^0)^*(X)=X^*X.
\end{equation}
Hence iterating $(\ppert\resl^0)^*$ on $X^*X$ gives an increasing sequence, which is, however, bounded by $\norm{X^*X}\idty$, because $\norm{(\ppert\resl^0)^*}\leq1$. Hence this sequence must have a weak limit, and because $(\ppert\resl^0)^*$ is normal, this limit is a fixed point.
Therefore $P$ and its transpose, and consequently $(\ppert\resl^0)$ must have a non-zero fixed point $\sigma$. But then the resolvent series for $\resl\sigma$ has all equal terms and hence diverges, in contradiction to the
trace estimate in the proof of Prop.~\ref{prop:minsol}.
\end{proof}

\subsection{Gauging and pure states in the domain }
The Kraus decomposition of a completely positive map is not unique, since it depends on the choice of a basis in the dilation space. Thus we may transform the jump operators linearly among each other by a unitary matrix without changing the generator. This corresponds to a basis change in the transit space $\transit$. In addition there is a change of Kraus operators of $\semg_t$ for small $t$, which mixes the $\sqrt tL_\alpha$ and $\idty+tK$. This is well-known in the bounded case, and is sometimes called a change of gauge. We will verify here that it survives mutatis mutandis in the unbounded case.

\begin{lem}\label{lem:gauge}
Let $K,L$ determine a standard generator as in \eqref{standUB}, and let $\lambda_\alpha\in\Cx$ with $\sum_\alpha\abs{\lambda_\alpha}^2<\infty$, and $\beta\in\Rl$. Then for $\phi\in\dom K$ set
\begin{eqnarray}
  L'_\alpha\phi &=& L_\alpha\phi+\lambda_\alpha\phi \label{Lgauge}\\
  K'\phi&=& K\phi+\sum_\alpha\overline{\lambda_\alpha}\,L_\alpha\phi+ \frac12\Bigl(i\beta+\sum_\alpha\abs{\lambda_\alpha}^2\Bigr)\phi \label{Kgauge}
\end{eqnarray}
Then the sum in the second term in \eqref{Kgauge} converges in norm. Moreover, $K'$ is a contraction generator with $\dom K'=\dom K$. The standard generators for $(K,L)$ and $(K',L')$ coincide on $\domsq K$, so that they determine the same minimal solution.
\end{lem}

\begin{proof}
First we show that $\norm{\sum_\alpha \overline{\lambda}_\alpha L_\alpha \phi}$ is $K$-bounded. Using the Cauchy-Schwarz inequality we have, for arbitrary $\psi\in\HH$,  $\phi\in\dom K$, and $\epsilon>0$
\begin{eqnarray*}
	\Bigl|{\sum_\alpha \overline{\lambda}_\alpha \braket \psi {L_\alpha \phi}}\Bigr|^2
    &\leq& \sum_i \abs{\overline{\lambda}_\alpha}^2 \sum_\alpha \abs{\braket \psi {L_\alpha \phi}}^2
	 \leq  A\, \norm{\psi}^2\, \sum_\alpha \norm{L_\alpha \phi}^2\\
	&\leq& A\, \norm{\psi}^2\, \abs{2\re \braket{\phi}{K\phi}}
      \leq \norm{\psi}^2\,4\Bigl(\frac A{2\epsilon}\norm\phi\Bigr)(\epsilon\norm{K\phi})\\
    &\leq& \norm{\psi}^2 \bigl(\epsilon\norm{K\phi}+ \frac A{2\epsilon}\norm\phi\bigr)^2
\end{eqnarray*}
where we have introduced the abbreviation $A=\sum_\alpha \abs{\overline{\lambda}_\alpha}^2$, used \eqref{standUB} at the second line,
and the estimate $4xy\leq(x+y)^2$ at the last.
Taking the square root and using that $\psi$ is arbitrary, we get
$\norm{\sum_\alpha \overline{\lambda}_\alpha {L_\alpha \phi}}\leq \epsilon \norm{K\phi}+ (A/(2\epsilon))\norm\phi$,
and, including the last term in \eqref{Kgauge}, $\norm{(K'-K)\phi}\leq\epsilon \norm{K\phi}+ C\norm\phi$, for some constant $C$.
That is, the perturbation is infinitesimally $K$-bounded. According to \cite[Theorem IV.1.1]{Kato}, $\epsilon<1$ is enough to conclude that $K'$ generates a semigroup with the same domain as $K$.

It remains to show that $K'$ is the generator of a contraction semigroup, i.e. that it is dissipative, which for a Hilbert  space operator just means $2\re\braket\phi{K'\phi}\leq 0$. For this we get
\begin{eqnarray*}
	2\re\braket\phi{K'\phi}&=& 2\re\braket\phi{K\phi}-2\re\sum_\alpha\braket{\lambda_\alpha\phi}{L_\alpha\phi}-\sum_\alpha\braket{\lambda_\alpha\phi}{\lambda_\alpha\phi}\\
	&=& 2\re\braket\phi{K\phi}+\sum_\alpha\norm{L_\alpha\phi}^2-\sum_\alpha\norm{L_\alpha\phi+{\alpha}_\alpha\phi}^2.
\end{eqnarray*}
Then the first two terms together are $\leq0$ because of \eqref{standUB}, and the third is obviously $\leq0$.

The equality of the generator then follows by the same elementary algebra as in the bounded case.
\end{proof}

A key result for the construction of non-standard generators is the following. It uses a condition from \cite{HolevoStoch}, which is related to the question whether the semigroup on the trace class is the dual of a semigroup on the compact operators. We show in Sect.~\ref{sec:NonClos} that it may be violated. On the other hand it is quite easy to verify in our two main examples.

\begin{prop}\label{prop:GENnoNewKetbra}
Let $\gen$ and $\gen_0$ be as in \eqref{standUB}. Assume in addition that each $L_\alpha$ is closable with $\dom L_\alpha^* \subset \dom K^*$ and $\sum_\alpha \norm{L_\alpha^* f}^2<\infty$ for $f\in \dom K^*$. \\
Then  $\ketbra \phi \psi \in \dom \gen$ for some $\phi,\psi\in\HH$  implies that  $\phi,\psi\in \dom K$.
\end{prop}

\begin{proof}
Following Theorem A.2 in \cite{HolevoStoch}, the semigroup satisfies the so called ``forward master equation'' with the generator
\begin{equation}\label{forward}
  \braket {f} {(\gen \omega) g}
     =\braket {K^* f}{\omega g}+\braket{f}{\omega K^* g}+\sum_\alpha \braket{L_\alpha^* f}{\omega L_\alpha^* g}
\end{equation}
for $\omega \in \dom \gen$.

Now let $\omega=\ketbra \phi \psi$ with $\phi,\psi$ not necessarily in $\dom K$, and pick a vector $g\in\dom K^*$ such that $\braket\psi g=1$. This is possible, because $\dom K^*$ is dense. Now we apply Lemma~\ref{lem:gauge} with $\lambda_\alpha=-\overline{\braket\psi{L_\alpha^*g}}$. This leads to an equivalent form of the generator, for which, however,
$\braket\psi{L_\alpha^*g}=0$. Therefore, \eqref{forward} simplifies to
\begin{equation}\label{sideways}
  \braket {f} {(\gen \omega) g} =\braket {K^* f}\phi\, \braket\psi g +\braket{f}\phi\braket\psi {K^* g}.
\end{equation}
Solving for the first term on the right, using $\braket\psi g=1$, we find
\begin{equation}
  \braket {K^* f}{\phi}=\Bigl\langle{f}\Bigm|{\gen(\omega) g}-\phi \braket {\psi}{K^* g}\Bigr\rangle.
\end{equation}
Therefore $\phi\in \dom K^{**}=\dom K$, and $K\phi={\gen(\omega) g}-\phi \braket {\psi}{K^* g}$. By the same argument applied to the hermitian conjugates we get $\psi\in\dom K$.
\end{proof}


\section{Examples of standard generators}
\subsection{Non-closable jump operators}\label{sec:NonClos}
A fundamental example of a contraction semigroup with unbounded generator is the half-sided shift on $\HH=\LL^2(\Rl^+,dx)$, given by
\begin{equation}\label{shift}
  \bigl(S_t\psi\bigr)(x)=\psi(x+t).
\end{equation}
Its generator $K$ is differentiation, so $\dom K$ consists of functions, which have an $\LL^2$-derivative. This means that they are, in particular, continuous, and hence, for $\psi\in\dom K$, the boundary value $\psi(0)$ is well-defined. This directly determines the exit space $\Exit=\Cx$ with $j\psi=\psi(0)$. Indeed,
\begin{equation}\label{exitshift}
  -\frac d{dt}\left.\braket{S_t\psi}{S_t\phi}\right|_{t=0}= -\frac d{dt}\int_t^\infty\!dx\ \overline{\psi(x)}\phi(x)=\overline{\psi(0)}\phi(0)=\braket{j\psi}{j\phi}.
\end{equation}
Hence the standard generators with no-event semigroup implemented by $S$ are parameterized by the cp map taking a one-dimensional system on exit $\Exit$ to the system Hilbert space, i.e., by a state $\Omega\in\th$. The intuitive picture is that whenever the system hits the boundary, it is reset to the ``rebound'' state $\Omega$.
The number of jump operators needed here depends on the mixedness of the rebound state $\Omega$. When $\Omega=\sum_\alpha\kettbra{\phi_\alpha}$ is the spectral resolution ($\phi_\alpha$ orthogonal but not normalized), we can set $L_\alpha:\Exit\to\HH$ to be $L_\alpha z=z\phi_\alpha$.

As operators on Hilbert space these jump operators are very ill-behaved. Formally, they would come out as $L_\alpha=\ketbra{\phi_\alpha}{\delta}$, where $\delta$ is the Dirac-$\delta$ at the origin. This $L_\alpha$ is not a closable operator, intuitively, because the value of a general $\LL^2$-function at a point is an ill-defined notion. More formally, we can find a sequence $\psi_n\in\dom K=\dom L_\alpha$ such that $\norm{\psi_n}\to0$, but $\psi_n(0)=42$. Then $L_\alpha\psi_n=42\phi_\alpha\neq0$, independently of $n$. Hence the closure of $L_\alpha$ would have to map $0$ into $42\phi_\alpha$, which is impossible for a linear operator. Since the usual definition of adjoint works well only for closable operators, the jump operators in the standard form \eqref{stand}, and even more so their adjoints, have to be interpreted with care. One can build a special notion of adjoint for this purpose \cite{ABaum}, but it is better to take the view of Prop.~\ref{prop:stdExit} and take $L_\alpha=M_\alpha j$, i.e., as completely determined by the {\it bounded} operators $M_\alpha$. In this way all the difficulties with singular $L_\alpha$ are controlled by the normalization loss of the no-event semigroup.
This is analogous to a well-known example of a generator perturbation for which the added term by itself makes little sense, namely point potentials ($\delta$-function potentials) for Schr\"odinger operators. Again, multiplication by a $\delta$-function, which is formally the potential ``added'' to the Laplacian, is a crazy operator by itself. However, as a perturbation of the Laplacian it makes sense and leads to a well-defined self-adjoint operator, which has an alternative description as the Laplacian with a modified boundary condition at the origin. The whole construction is quite stable, and we can also obtain the perturbed operator as the strong resolvent limit of Schr\"odinger operators with suitably scaled potentials with small support around the origin.

The example of this section is also discussed in \cite{Holx}, where it is shown that Arveson's ``domain algebra'' \cite{ArvDomAlg} can be trivial.

\subsection{Quantum birth process}
\subsubsection{The process}
A standard example of the classical theory is the so-called pure birth process. The state of the system at any time is given by an integer $n$, from where it can jump to $n+1$ with rate $\mu_n>0$. The generator thus acts on $\rho\in\ell^1(\Nl)$ as
\begin{equation}\label{pbirth}
  (\gen\rho)(n)=
      \begin{cases}\mu_{n-1}\rho(n-1)-\mu_n\rho(n)& \quad\mbox{for}\ n>0\\
                     \strut\hskip60pt -\mu_0\rho(0) &\quad\mbox{for}\ n=0.
      \end{cases}
\end{equation}
The case distinction can be avoided by the convention $\rho(-1)=0$. By telescoping sum one verifies $\sum_n(\gen\rho)(n)=0$, so the process appears to be conservative. On the other hand, noting that the expected time for the transition from $n$ to $n+1$ is $\mu_n\inv$ it seems possible that the process reaches infinity in finite time when $\mu_n$ increases sufficiently rapidly, i.e.,
\begin{equation}\label{birthXplode}
  \sum_n\frac1{\mu_n}=\tau<\infty.
\end{equation}
Indeed this is part of the well-established lore on this process (see \cite[Sect.~XVII.4]{Feller} and below). Our interest here is in a closely related quantum process, which is a standard semigroup on $\HH=\ell^2(\Nl)$ with $K$ and a single jump operator $L$ given by
\begin{eqnarray}
  K\ket n&=&-\frac{1}{2}\mu_n\ket n, \quad \dom K=\Bigl\{ \psi \in \ell^2(\Nl): \sum_{n=0}^{\infty}
            \mu_n^2\abs{\brAket{\psi}{n}}^2<\infty\Bigr\}, \nonumber \\
   L\ket n&=&\sqrt{\mu_n}\ket{n+1}, \quad \dom L \subset\dom K, \nonumber
\end{eqnarray}
where $\{\ket n\}$ is the canonical basis of the Hilbert space. As usual, we denote by $\gen^0\rho=K\rho+\rho K^*$ the no-event generator, which corresponds to the first term in the expression for the standard generator
\begin{equation}\label{QBirthrho}
  \bra n\gen\rho\ket m=-\frac12(\mu_n+\mu_m)\bra n\rho\ket m+ \sqrt{\mu_{n-1}\mu_{m-1}}\ \bra{n-1}\rho\ket{m-1}.
\end{equation}
This is the quantum analogue of \eqref{pbirth}, a simplified and generalized version of a process first studied in \cite[Example~3.3]{DaviesN}.  It reduces precisely to the classical case for purely diagonal density operators.  We therefore call the process generated by $K$ and $L$ the {\it quantum birth process}. Like its classical counterpart it is formally conservative, but due to the possibility of escape to infinity it may actually fail to be conservative. It will then be interesting to look at the details of the escape: Is there any quantum information ``coherently'' pushed to infinity?

For this simple example the resolvent series \eqref{resmin} can be summed explicitly. We get,
for any  $\rho\in\th$,
\begin{eqnarray}  \label{QBreso}
  \bra n\resl\rho\ket m
    &=& \frac1{\lambda+\frac12(\mu_n+\mu_m)}\sum_{k=0}^{\min(n,m)}p^k_{nm}\ \bra{n-k}\rho\ket{m-k}\\
   p^k_{nm}&=& \prod_{j=1}^k\frac{\sqrt{\mu_{n-j}\mu_{m-j}}}{\lambda+\frac12(\mu_{n-j}+\mu_{m-j})}.
\end{eqnarray}
Thus the domain of the generator of the minimal solution is $\dom\gen=\{\resl\rho'|\rho'\in\th\}$, and $\gen\rho=\gen\resl\rho'=\lambda\resl\rho'-\rho'$.

In general, it is not easy to determine $\dom\gen$ from the expression \eqref{stand}, here \eqref{QBirthrho}, which merely expresses the generator on the domain $\domsq K$. On the other hand, the matrix elements on the right hand side of \eqref{QBirthrho} make sense for any bounded operator $\rho$. It turns out that this reading of \eqref{QBirthrho} correctly expresses the extension by minimal solution:

\begin{lem}\label{lem:QBok}
For $\rho\in\dom\gen$, and all $n,m\in\Nl$, Eq.~\eqref{QBirthrho} holds. Conversely, if, for some trace class operator $\rho$,
the right hand side of Eq.~\eqref{QBirthrho} gives the matrix elements of a trace class operator, then $\rho\in\dom\gen$.
\end{lem}

\begin{proof}
Both \eqref{QBirthrho} and \eqref{QBreso} involve finite sums only for fixed $n,m$. Therefore, we can can consider them to define extensions $\gen^\sharp$ and $\resl^\sharp$ of $\gen$ and $\resl$ to arbitrary matrices $\rho$. It is straightforward to verify that $\gen^\sharp\resl^\sharp=\lambda\resl^\sharp-\id^\sharp=\resl^\sharp\gen^\sharp$. Take the first equation, and apply it to some $\rho'\in\th$. This shows that
$\gen^\sharp\resl\rho'=\gen^\sharp\resl^\sharp\rho'=\lambda\resl^\sharp\rho'-\rho'=\lambda\resl\rho'-\rho'=\gen\resl\rho'$, i.e., $\gen^\sharp$ and $\gen$ coincide on $\dom\gen$.

Now suppose that $\rho$ and $\gen^\sharp\rho$ are both trace class. Then by the second equation $\rho=\resl^\sharp(\lambda\rho-\gen^\sharp\rho)\in\resl^\sharp\th=\resl\th=\dom\gen$.
\end{proof}

\subsubsection{Conservativity}
From the integral \eqref{resgen} one sees that $\semg_t$ is conservative if and only if
$\tr\lambda\resl\rho=\tr\rho$ for all $\rho$. The trace of \eqref{QBreso} depends only on the sums with $n=m$, and hence the conservativity is exactly the same as for the classical problem.
The resolvent actually contains more information. Let $m(t)=-d/(dt)\tr\semg_t\rho$ be the ``arrival probability density'' at infinity. Then its Laplace transform is
\begin{equation}\label{arrivalatinfty}
 \widehat m(\lambda)=\int_0^\infty\!\! dt\ e^{-\lambda t}\ m(t)= 1-\tr\lambda\resl\rho.
\end{equation}
Starting from $\rho=\kettbra n$, and introducing the abbreviation $c_\alpha=\mu_{n+j}/(\lambda+\mu_{n+j})$  we get from \eqref{QBreso}
\begin{eqnarray}\label{QBcons}
  \widehat m(\lambda)&=&1-\tr\lambda\resl\kettbra n
      = 1- \sum_{k=0}^\infty \frac\lambda{\lambda+\mu_{n+k}}\prod_{j=0}^{k-1}\frac{\mu_{n+j}}{\lambda+\mu_{n+j}}\nonumber\\
     &=&1- \sum_{k=0}^\infty (1-c_k)\prod_{j=0}^{k-1}c_\alpha=\lim_{N\to\infty}\prod_{j=0}^{N}c_\alpha \nonumber\\
     &=& \prod_{j=n}^{\infty}\frac1{1+\lambda\mu_\alpha\inv}.
\end{eqnarray}
This has a straightforward probabilistic interpretation: The probability density of a sum of independent random variables is the convolution of the individual densities, corresponding to the product of the Laplace transforms. Hence the ``arrival time at infinity'' is the sum of infinitely many independent contributions, each exponentially distributed with density $\mu_\alpha e^{-\mu_\alpha t}$. When $\tau=\sum_\alpha\mu_\alpha\inv=\infty$, this sum is actually infinite with probability $1$, and $\widehat m(\lambda)=0$.

\subsubsection{Domain increase}
Next we consider the question whether the inclusion  $\dom\gen\supset\dom\gen^0$ is strict. For this it is helpful to note that for any $\rho\in\dom\gen$ and $q\in\Ir$ the limit
\begin{equation}\label{QRlf}
  \Phi_q(\rho)=\lim_{n\to\infty}\frac12(\mu_n+\mu_{n+q})\brAAket n\rho{n+q}
\end{equation}
exists. Indeed, setting $\rho=\resl\rho'$ this is clear from \eqref{QBreso}, using $p^k_{nm}\leq1$ and $\rho'\in\th$. Moreover, if $\rho\in\dom\gen^0$ the matrix elements in the above limit belong to the trace class operator $\gen^0\rho$, and therefore are summable and have to go to zero, so $\Phi_q(\dom\gen^0)=\{0\}$. We note that $\Phi_0$ plays a special role since, for $\rho\in\dom\gen$,
\begin{equation}\label{QBPhi0}
  \Phi_0(\rho)=\lim_{n\to\infty} \sum_{m=0}^n \Bigl(\mu_m\brAAket m\rho m-\mu_{m{-}1}\brAAket{m{-}1}\rho{m{-}1}\Bigr)=-\tr\gen\rho
\end{equation}
is exactly the infinitesimal normalization loss. When the semigroup is not conservative (the only case we consider now), we can directly find an element on which this does not vanish:
\begin{equation}\label{QBsig}
  \sigma=\sum_n\frac1{\mu_n}\kettbra n, \qquad\mbox{with}\ \Phi_0(\sigma)=1.
\end{equation}
For the other values of $q$ the existence of such elements depends, in fact, on how fast the $\mu_n$ grow.

\begin{prop} Let the rates $\mu_n$ {\bf grow moderately} in the sense that, for all $q,n$,
\begin{equation}\label{moderate}
 \left|1- \frac{\mu_{n+q}}{\mu_{n}}\right| \leq \frac cn
\end{equation}
for some constant $c$ independent of $n$. Then, for any $q\in\Ir$, let
\begin{equation}\label{sigf}
   \sigma^q= \sum_n \frac2{\mu_n+\mu_{n+q}} \ketbra n{n+q}.
\end{equation}
Then $\sigma^q\in\dom\gen$, and $\Phi_q(\sigma^{q'})=\delta_{qq'}$.
\end{prop}

Moderate growth covers rational functions, stretched exponentials like $\mu_n\sim\exp(a n^\alpha)$ with $\alpha<1$, but exponentials $\mu_n=e^{an}$ clearly do not satisfy this condition.

\begin{proof}
The matrix \eqref{sigf} is clearly positive definite, and $\tr\sigma^q=\tau <\infty$. The critical question is whether $\gen\sigma^q$, as defined by \eqref{QBirthrho} is trace class. Like $\sigma^q$ itself, $\gen\sigma^q$ is of the form $\sum_na_n\ketbra n{n+q}$, and such an operator is trace class iff $\sum_n\abs{a_n}<\infty$ (Think of this as a diagonal operator multiplied with a shift from one side). Thus we have to show that the sum
\begin{equation}\label{QBsumtc}
  \sum_n\abs{\brAAket n{\gen\sigma^q}{n+q}}=\sum_n\Bigl|-1+\frac{2\sqrt{\mu_{n-1}\mu_{n+q-1}}}{\mu_{n-1}+\mu_{n+q-1}}\Bigr|
\end{equation}
is finite. Introducing the function
\begin{equation}\label{gab}
  g(a,b)=1-\frac{2\sqrt{ab}}{a+b}=\frac{(\sqrt a-\sqrt b)^2}{a+b}\leq \left(1-\frac ba\right)^2
\end{equation}
we find that for moderately growing $\mu_n$ the terms in the sum \eqref{QBsumtc} are bounded by $(c/n)^2$, so the sum converges.
\end{proof}

\noindent{\it Example: Exponentially growing $\mu_n$}\\
Let us put $\mu_n=a^n$ for some $a>1$. Then, for $q\neq0$,  the sum \eqref{QBsumtc} has all equal terms, and hence diverges. While the limit \eqref{QRlf} still exists for $\rho=\sigma^q$, and is equal to $1$, this does not help to establish domain increase, because $\sigma^q\notin\dom\gen$. Nor is there any other choice of $\rho$ of which we can prove in this way that $\rho\in\dom\gen\setminus\dom\gen^0$: For {\it any} $\rho\in\dom\gen$ we get $\Phi_q(\rho)=0$.

\begin{proof}
Consider the resolvent sum \eqref{QBreso}. Each factor in $p^k_{n,m}$ with $m=n+q$ is
\begin{equation}\label{QRestimexp}
  \frac{\sqrt{\mu_{n-j}\mu_{m-j}}}{\lambda+\frac12(\mu_{n-j}+\mu_{m-j})}
   \leq \frac{2\sqrt{\mu_{n-j}\mu_{m-j}}}{\mu_{n-j}+\mu_{m-j}}
   \leq \frac{2a^{q/2}}{1+a^q}=:\gamma.
\end{equation}
Hence $p^k_{n,n+q}\leq \gamma^k$, which is summable with respect to $k$. Assuming $q\geq1$ without loss,
\begin{eqnarray}\label{phiRapid}
    \abs{\Phi_q(\resl\rho)}
      &=&\lim_n\frac{\frac12(\mu_n+\mu_{n+q})}{\lambda+\frac12(\mu_n+\mu_{n+q})}
                   \left|\sum_kp^k_{n,n+q}\brAAket n\rho{n+q}\right|\nonumber\\
      &\leq&\lim_n\sum_{k=0}^{n} \gamma^k r_{n-k},
\end{eqnarray}
where we abbreviated $r_n=\abs{\brAAket n\rho{n+q}}$. This is a summable sequence because $\rho$ is trace class. The sum is consequently summable as the convolution of two such sequences, and therefore goes to zero as $n\to\infty$.
\end{proof}

\subsubsection{No new pure states}
We have seen that $\dom\gen$ is properly larger than $\dom\gen^0$. But are there also additional pure states in this larger domain? We could use Prop.~\ref{prop:GENnoNewKetbra} to answer this in the negative. Instead we give a simple alternative argument based on the range of resolvents.

\begin{prop}\label{prop:noNewKetbra}
 Let $\gen$ and $\gen^0$ be as above and $\ketbra\phi\psi\in\dom\gen$. Then $\ketbra\phi\psi\in\dom\gen^0$, i.e., $\phi,\psi\in\dom K$.
\end{prop}

\begin{proof}
Since $\ketbra\phi\psi\in\dom\gen$ we may write $\ketbra\phi\psi=\resl\rho$ for some $\rho\in\th$. Let $m$ be the smallest index for which $\rho\ket m\neq0$. Then in the formula \eqref{QBreso} for the resolvent only the term $k=0$ gives a non-zero contribution. Noting that $p^0_{nm}=1$ we get
\begin{eqnarray}
  \bra n\resl\rho\ket m&=&\frac1{\lambda+\frac12(\mu_n+\mu_m)}\bra n\rho\ket m\quad\mbox{for all\ }n\nonumber\\
  \resl\rho\ket m&=&(\lambda+\mu_m/2-K)\inv\rho\ket m.\nonumber\\
  \phi\,\braket\psi m&\in&(\lambda+\mu_m/2-K)\inv\HH=\dom K.\nonumber
\end{eqnarray}
Now $\brAket\psi m$ cannot vanish, because $\rho\ket m\neq0$. Hence $\phi\in\dom K$, and the same argument applied to $\rho^*$ gives also $\psi\in\dom K$.
\end{proof}

\subsection{Diffusion on diagonals}
This is the basis for the example of a non-standard generator given in \cite{HolevoExa}. The basic idea is very similar to the quantum birth process, and the main conclusion is the same. However, the presentation in \cite{HolevoExa} was rather sketchy and incomplete, and did not mention an argument along the lines of Prop.~\ref{prop:GENnoNewKetbra}. These clarifications were the focus of our collaboration, and have been independently summarized in \cite{Holx}. Therefore we can be brief here.

The system Hilbert space in this case is $\HH=\LL^2(\Rl_+,dx)$. In order to stress the analogies we use the same notations as above for the generators. They are
\begin{eqnarray}
  K &=& \frac{d^2}{dx^2}\qquad \quad \dom K = \bigl\{\psi\in\HH\bigm| \psi(0)=0,\ \psi''\in\HH\bigr\} \\
  L &=& \sqrt2\,\frac{d}{dx}\qquad  \dom L=\dom K.
\end{eqnarray}
$K$ generates a diffusion with absorbtion at the boundary point $0$. Similarly, when seen as acting on integral kernels $\rho(x,y)$, $\gen$ generates a diffusion with a degenerate diffusion operator $(\frac d{dx}+\frac d{dy})^2$, corresponding to diffusion along the diagonals $x-y=const$ with absorption at the boundary of the positive quadrant. Both semigroups can be solved explicitly by the reflection trick: The semigroup without the absorbing boundary condition is translation invariant, and acts by convolution with a Gaussian kernel. The solution with absorption is then obtained by first extending the initial function to an antisymmetric one on the whole line, applying the Gaussian kernel and restricting to the half line afterwards.

In this way we get the time evolution (see \cite{Holx}, correcting \cite{HolevoExa}), written in terms of its action on integral kernels
$\omega:\Rl_+\times\Rl_+\to\Cx$ representing trace class operators:
\begin{eqnarray}
  (\semg_t\omega)(x,y)&=& \frac1{2\sqrt{\pi t}}\int_0^\infty\!\!d\xi\ \sum_{n=0,1}(-1)^n\exp\left\lbrace-\frac1{4t}\bigl|\min(x,y)-(-1)^n\xi\bigr|^2\right\rbrace \nonumber\\
                      &&\strut\qquad\times \omega(\xi+[x-y]_+,\xi+[y-x]_+) \label{HolSol}
\end{eqnarray}
Here we wrote $x_+=\max\{0,x\}$ for the positive part of a number. By integration \eqref{resgen} we get the resolvent
\begin{eqnarray}\label{Res}
  (\resl\omega)(x,y)&=&\int_0^\infty\!\! d\xi\ f_\lambda^{x,y}(\xi)\ \omega\bigl(\xi+[x-y]_+,\xi+[y-x]_+\bigr)\\
  f_\lambda^{x,y}(\xi)&=& \frac1{2\sqrt\lambda} \sum_{n=0,1}(-1)^n\exp\left\lbrace-\sqrt\lambda\bigl|\min(x,y)-(-1)^n\xi\bigr|\right\rbrace.\label{fxy}
\end{eqnarray}
Since $f_\lambda^{0,y}=f_\lambda^{x,0}=0$ we must have $\resl\omega(0,y)=\resl\omega(x,0)=0$ for all $\omega$. Hence, for all $\omega\in\dom\gen$,
$\omega(0,y)=\omega(x,0)=0$. Similarly, one sees that the kernel $\omega(x,y)$ has to be continuous for $\omega\in\dom\gen$.

To find the normalization loss we can integrate \eqref{HolSol} to get
\begin{equation}\label{HolNormalize}
  \tr\semg_t\omega=\tr\omega-\int_0^\infty\!\!d\xi\ \erfc\left(\frac\xi{2\sqrt t}\right) \omega(\xi,\xi),
\end{equation}
where $\erfc$ denotes the complementary error function. We substitute $\xi\mapsto 2\sqrt t \eta$ and take from \cite{Holx} the information that, for $\omega\in\dom\gen$, we have $\omega(x,x)=\Lambda x+\order(x)$ as $x\to0$.
Then by dominated convergence, and using $\int_0^\infty\!dx\ x\,\erfc(x)=1/4$, we find
\begin{equation}\label{HolNormalize2}
  \tr\semg_t\omega=\tr\omega- 2\sqrt t \int_0^\infty\!\!d\eta\ \erfc(\eta)\omega(2\sqrt t\,\eta,2\sqrt t\,\eta)
        =\tr\omega-t\Lambda\ +\order(t)
\end{equation}
The diagonal derivative $\Lambda=-d/(dx)\omega(x,x)|_{x=0}$ plays the same role as $\Phi_0(\rho)$ in the previous section (compare \eqref{QBPhi0}).

The crucial observation is once again, that $\ketbra\phi\psi\in\dom\gen$ implies $\ketbra\phi\psi\in\dom\gen^0$. Two techniques are available for showing this: In analogy to Prop.~\ref{prop:noNewKetbra} one can
directly show that $(\resl\omega)\chi\in\dom K$ for suitable $\chi$. But in this case it is preferable to invoke Prop.~\ref{prop:GENnoNewKetbra}.

\section{Examples of non-standard generators}\label{sec:nonstand}
We focus here on the examples, which come immediately out of the two examples studied in the previous section: The quantum birth and the diagonal diffusion semigroups. In both cases we considered a standard generator $\gen$, arising from positive perturbation of a no-event generator $\gen^0$.  Now we go one step further and add to $\gen$ another positive term, leading to the generator $\dom\genh$ of a conservative semigroup. This perturbation again follows the minimal solution pattern (Sect.~\ref{sec:minsol}) with a rank one perturbation, for simplicity. That is,  we set
\begin{equation}\label{genh}
  \genh\rho=\gen\rho-\tr(\gen\rho)\rhoh,\qquad\dom\genh=\dom\gen.
\end{equation}
The added term is completely positive on $\dom\gen$, because normalization loss is negative. The equality of domains follows from Lemma~\ref{lem:finrank}. In dynamical terms, the process will reset to $\rhoh$, whenever there is an ``arrival event'', which under $\gen$ would mean a loss of normalization: in the quantum birth case, this will be an arrival at infinity, and in the diagonal diffusion case an arrival at the origin.

We have here two construction steps in which a completely positive term is added to the generator. Why can they not be fused into a single step adding both terms simultaneously? Indeed, if this were possible, $\genh$ would be, by definition, a standard generator. The key observation is that $\gen$ is infinitesimally trace preserving on $\dom\gen^0$, so $\ppert=\gen-\gen^0$ is already as large as it can be. However, since the semigroup $\exp(t\gen)$ is {\it not} conservative, $\dom\gen$ must be properly larger than $\dom\gen^0$, because a generator which is infinitesimally conservative on its full domain would generate a conservative semigroup. The term added when passing from $\gen$ to $\genh$ vanishes on $\dom\gen^0$, so is strictly associated with the ``new'' part of the domain. The various generators and domains are graphically summarized in Fig.~\ref{fig:doms}.
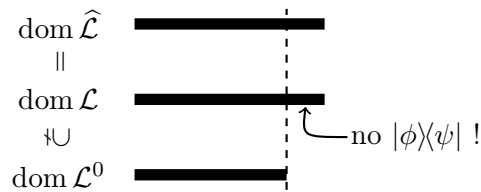
\begin{figure}[h]
\centering
\begin{tikzpicture}\def\w{.15cm}
   \node at (-1,0) {$\dom\gen^0$};    \draw[line width=\w] (0,0)--(2,0);
   \node at (-1,1) {$\dom\gen$};    \draw[line width=\w] (0,1)--(2.5,1);
   \node at (-1,2) {$\dom\genh$};    \draw[line width=\w] (0,2)--(2.5,2);
   \draw[thick,dashed] (2,-.2)--(2,2.2);
   \node[rotate=-90] at (-1,0.5) {$\supsetneq$}; \node[rotate=-90] at (-1,1.5) {$=$};
   \draw[thick,->] (2.8,0.5) .. controls (2.2,0.5).. (2.25,.9);
   \node at (3.7,0.5) {no $\ketbra\phi\psi$\ !};
\end{tikzpicture}
\caption{Generators and their domains in the construction of a non-standard generator $\genh$.}
\label{fig:doms}
\end{figure}

The same relations may hold in the discrete classical case, namely whenever $\gen$ is a standard generator which appears conservative on the pure states in its domain, but actually allows some escape to infinity and hence generates a non-conservative semigroup. Indeed, any standard generator is completely determined by its action on the pure states (even though its full domain might not be spanned just by these). If $\genh$ were standard, since it coincides with $\gen$ on the pure states, we would have $\genh=\gen$. On the other hand, these generators are clearly different, since one generates a conservative semigroup and the other does not.

In the quantum case this argument is too simple, since not all pure states are in the domain, but only those $\kettbra\psi$ with $\psi\in\dom K$. So the possibility we have to discuss is that there might be another contraction generator $\widetilde K$, and associated no-event semigroup generator $\widetilde\gen^0$, from which $\genh$ arises in a one-step minimal solution construction.

It is here that we can use the fact (Props.~\ref{prop:noNewKetbra} and \ref{prop:GENnoNewKetbra}) that for all $\ketbra\phi\psi\in\dom\genh$ we actually have $\phi,\psi\in\dom K$. So even if we had started from some other $\widetilde K$, we could still reconstruct  $\dom K$ from $\dom\genh=\dom\gen$. Since $\genh$ and $\gen$ coincide on $\domsq K$  they would arise as minimal solutions from the same equation on $\domsq K$.  Therefore, in both examples $\genh$ is non-standard.

\section{Conclusions and outlook}
We have explicitly defined a notion of unbounded GKLS generators, which we feel summarizes an agreement in the literature \cite{DaviesN,DaviesG,FagnolaProy,HolevoExc,HolevoStoch}. However, as we have shown, not all generators are of this form.

As the defining feature of the standard form we took the existence of many pure states in the domain.
Alternatively,  one can start from the observation that $\semg_t-\semg^0_t$ is completely positive for all $t$. Bill Arveson \cite{ArvesonB,ArvDomAlg} calls no-event semigroups $\semg_t^0$ with this property the {\it units} of $\semg_t$. His ``standard'' case, which he calls ``type I'', is defined by the existence of many such units, which arise by Lemma~\ref{lem:gauge} from each other. He gives examples, which are not of this kind, especially one with no units at all \cite{ArvesonB,Powers} (``type III''). It is unclear how our notion of standardness and Arveson's type I are exactly related, and this seems like an excellent question for further research.

Of course, the long-term goal is to arrive at a better understanding of non-standard generators, for which a closer look at the classical theory will certainly be helpful.

\section*{Acknowledgements}
I.S. and R.F.W. are supported by RTG 1991 and SFB DQ-mat of the DFG.
\enlargethispage{2cm}

\bibliography{gen}
\end{document}